\documentclass{aip-cp}
\usepackage[numbers]{natbib}
\usepackage{rotating}
\usepackage{graphicx}
\usepackage[numbers]{natbib}
\usepackage{rotating}
\usepackage{eqnarray}
\usepackage{savesym}
\savesymbol{iint}
\savesymbol{iiint}
\savesymbol{iiiint}
\savesymbol{idotsint}
\savesymbol{openbox}
\savesymbol{syrv}
\usepackage{amsmath,amsthm}
\usepackage{amssymb,euscript}
\usepackage[cp1251]{inputenc}
\usepackage{lmodern}
\usepackage[english]{babel}
\usepackage{siunitx}
\usepackage[T5,T1]{fontenc}
\newcommand{\llll}[1] {\left #1}
\newcommand{\rrrr}[1] {\right #1}

\newcommand{\dddd}[2]{\frac{#1}{#2}}
\newcommand{\pppp}{\partial}
\newcommand{\aaaa}{\alpha}
\newcommand{\tttt}{\tau}
\newcommand{\vvvv}{\varphi}

\newcommand{\ssss}{\sigma}

\newcommand{\GGGG}{\Gamma}

\newcommand{\llllll}{\lambda}
\newtheorem{lem}{Lemma}

\begin{document}

\title {High Order Finite Difference Schemes on Non-Uniform Meshes for the  Time-Fractional Black-Scholes Equation}
\author{Yuri  M.   Dimitrov}
\author{Lubin G. Vulkov\\
FNSE, University of Rousse, Rousse 7017, Bulgaria\\
ymdimitrov@uni-ruse.bg, lvalkov@uni-ruse.bg}
\maketitle

\begin{abstract}
We construct a three-point compact finite difference scheme on a  non-uniform mesh for the time-fractional Black-Scholes equation. We show that for special graded meshes used in finance, the Tavella-Randall and the quadratic meshes the numerical solution has a fourth-order accuracy in space. Numerical experiments are discussed.
\end{abstract}
 \section{Introduction}
The Black-Scholes-Merton model for option prices is an important model in financial mathematics. Since its discovery in the early seventies, it has been widely used in practice and has been studied rigorously using analytical and computational methods.
The value of an option, denoted by $V$,  depends on the current market value of the underlying asset $s$, and the remaining time $t$ until the option expires:
$V=V(s,t)$. The    Black-Scholes equation (BS) is a backward-in-time parabolic equation \cite{BlackScholes}
\begin{equation}\label{BS}
LV:=\frac{\partial V} {\partial t } +\frac{1}{2}\sigma^{2} s^{2}
\frac{\partial^{2} V} {\partial s^{2} }  +
(r-d)s\frac{\partial V} {\partial s } -rV =0 ,
\end{equation}
where $\sigma$ is the annual volatility of the asset price, $r$ is the
risk-free interest rate, $d$ is the dividend yield and $T$ is the   expiry date ($t=0$ means "today"). Due to the complexity of the financial markets, a number of improvements and modifications to the model  have been proposed in order to improve the its accuracy depending on the state of the market.  The change in the option price with time in the fractional model for option prices is a fractional transmission system. This assumption implies that the total flux rate of the option price
${\overline{Y}}(s,t)$ per unit time from the current time $t$ to the expiry date $T$ and the option price
$V(s,t)$ satisfy
\begin{equation}\label{2}
\int_{t}^{T} {\overline{Y}} (s,t')dt' =s^{d_{f} - 1}\int_{t}^{T}H (t'-t)
[V(s,t')-V(s,T)] dt',
\end{equation}
where $H(t)$ is the transmission functional and $d_{f}$ is the Hausdorff dimension of the
fractional transmission system. As pointed in \cite{LiangWangZhang},  the essence of \eqref{2} is a conservation equation containing an explicit reference to the history of the diffusion process of the
option price on a fractal structure. We further assume that the diffusion sets are underlying  fractals and the transmission function
$H(t) = \frac{ A_{\alpha} } {    \Gamma (1-\alpha) t^{  \alpha } }$, where
$A_{\alpha}$ and $\alpha$ are constants and $\aaaa$ is the transmission exponent. Now, by differentiating   (2) with respect to $t$, we obtain
\begin{equation}
{\overline{Y}} (s,t) =s^{d_{f} - 1}\frac{d}{dt}\int_{t}^{T}H (t'-t)[V(s,t')-V(s,T)] dt'.
\end{equation}
\noindent On the other hand, from the BS equation, we have
$$
{\overline{Y}} (s,t) =\frac{1} {2}
\sigma^{2} s^{2}
\frac{  \partial^{2} V} {    \partial s^{2} }  +
(r-d) s\frac{  \partial V} {    \partial s } -r V ,
$$
\noindent which combined with (3), yields \cite{LiangWangZhang}
\begin{equation}\label{FM}
A_{\alpha}s^{d_{f} - 1}
\frac{\partial^{\alpha} V} {\partial t^{\alpha} } +
\frac{1}{2} \sigma^{2} s^{2}
\frac{\partial^{2} V}{\partial s^{2} }  +(r-d) s
\frac{\partial V}{\partial s } -r V =0,
\end{equation}
where $ \frac{\partial^{\alpha} V} {\partial t^{\alpha} } $ is the modified Riemann-Liouville derivative defined as
\begin{equation*}
\frac{  \partial^{\alpha} V} {\partial t^{\alpha} } =
\frac{1}{\Gamma (n-\alpha) }
\frac{   \partial^{n} }{  \partial t^{n} }
\int_{t}^{T}\frac{V(s,t')-V(s,T)}{(t'-t)^{1+\aaaa-n}}dt'  \quad
 \mbox{for} \quad
 n-1 \leq \alpha < n.
\end{equation*}
When $\alpha =1$ and under natural conditions for the function $V(s,t)$ the modified Riemann-Liouville derivative $\frac{  \partial^{\alpha} V} {\partial t^{\alpha} }$ is equal to the  partial derivative
$\frac{  \partial V} {\partial t }$ and  $\frac{  \partial^{\alpha} V} {\partial t^{\alpha} }$ is equal to the  Caputo derivative, when $0<\aaaa<1$ \cite{ChenXuZhu}
\begin{equation*}
\frac{  \partial^{\alpha} V} {\partial t^{\alpha} } =
\frac{1}{\Gamma (1-\alpha) }
\frac{   \partial }{  \partial t }
\int_{t}^{T}\frac{V(s,t')-V(s,T)}{(t'-t)^{\aaaa}}dt' =
\frac{1}{\Gamma (1-\alpha) }
\int_{t}^{T}\frac{V_t(s,t')}{(t'-t)^{\aaaa}}dt'.
\end{equation*}
Therefore equation (4) transforms to   (1)  when $A_{\alpha}=d_{f} =1$ and $\aaaa=1$. For consistency with
 the benchmark Black-Scholes model, following \cite{LiangWangZhang}, we assume that $A_{\alpha}=d_{f} =1$. In fact, the compact difference approximation \eqref{DA} described below, can be easily extended to other values of $A_{\alpha}$ and $d_{f}$.

In the last decade, a great deal of effort has been devoted to developing high-order compact schemes, which utilize the grid nodes directly adjacent to the central nodes.   Three-point compact finite-difference schemes on  uniform  spacial meshes for the time-fractional advection-diffusion equation are constructed in \cite{GaoSun}. The non-uniform meshes improve the efficiency of the numerical solutions of equation (4), which has a second order degeneration at $s=0$  \cite{TavellaRandall,HaentjensHout,SweilamRizkAbouHasan,BodeauRibouletRoncalli}. The goal of the present paper is to construct a high-order three-point compact finite-difference scheme for the time-fractional Black-Scholes (TFBS) equation \eqref{TFBS} and the time-fractional Black-Scholes  equation \eqref{TFBSD} in diffusion form (TFBSD) on a spacial non-uniform mesh. The outline of the paper is as follows. In section 2, we introduce and analyze a fourth-order compact approximation \eqref{CA} for the
second derivative on a non-uniform mesh.   In  section 3 we use approximation \eqref{CA} to construct a compact finite-difference scheme for the TFBSD equation on  special non-uniform meshes used in finance and we present the results of the numerical experiments for  test examples.
\section{Compact approximation on a non-uniform mesh}
Non-uniform grids are frequently used for numerical solution of differential equations, especially for equations with singular solutions, in order to improve the accuracy of the numerical method. The most commonly used grid in finance is the Tavella-Randall grid, which resolves the effect of the singularity of the initial condition of the BS equation at the striking price $s=K$. Let  $\varphi (x) $ be an increasing function on the interval $[0,1]$  with values $\vvvv(0)=s^{-}$ and $\vvvv(1)=s^{+}$.
Denote by  $\mathcal{M}^N=\llll\{x_n=n h\rrrr\}_{n=0}^N$ the uniform net on the interval $[0,1]$, where $h=1/N$ and $N$ is a positive integer.
We use the function $\vvvv$ to define non-uniform meshes $\mathcal{M}^N_\vvvv$ on the interval $\llll[s^{-},s^{+}\rrrr]$ by
$$\mathcal{M}^N_\vvvv=\llll\{ s_n=\vvvv (x_n) | n=0,1,\cdots, N   \rrrr\}.
$$
The mesh $\mathcal{M}^N_\vvvv$ has non-uniform  mesh steps  $ h_{n} = s_{n+1} - s_{n}$. When the function $\vvvv$ is a differentiable function with a bounded first derivative we determine a bound on the mesh steps, using the mean value theorem
$$h_n=s_{n+1}-s_n=\vvvv(x_{n+1})-\vvvv (x_n)=h \vvvv'(y_n),$$
where $y_n\in\llll(x_n,x_{n+1}\rrrr)$. The maximal length of the subintervals of the mesh $\mathcal{M}^N_\vvvv$ is bounded by the maximal value of the first derivative of the function $\vvvv$
$$h_n\leq \llll(\max_{x\in [0,1]} \vvvv'(x) \rrrr)h.$$

\subsection{Non-Uniform Grids in Finance}
The Black-Scholes equation is an important equation for practical applications and its numerical and analytical solution is an active research topic. The computation of the numerical solution of the BS equation is an interesting problem  because of the singularities of the equation and its non-smooth initial condition. Non-uniform grids for numerical solution of the BS equation are used \cite{TavellaRandall,HaentjensHout,SweilamRizkAbouHasan,BodeauRibouletRoncalli,WangZhangFang}, in order to overcome the deficiencies of the numerical solutions  at the points $s=0$ and $s=K$. In this paper we discuss a fourth-order accurate three-point compact difference approximation for the TFBSD equation on the Tavella-Randall and the quadratic non-uniform grids.

$\bullet$ {\it Tavella-Randall non-uniform grid $\mathcal{M}^{TR}_{\llllll,K}$ on the interval $\llll[s^{-},s^{+}\rrrr]$}
$$s_n=s^*+\llllll \sinh \llll( c_1\llll(1- \dddd{n}{N}\rrrr)+c_2 \dddd{n}{N}
\rrrr),\qquad \text{where}\quad s^{-}<s^*<s^{+},$$
$$ c_1=\sinh^{-1}\llll(\dddd{s^{-}-s^{*}}{\llllll}\rrrr),\quad c_2=\sinh^{-1}\llll(\dddd{s^{+}-s^{*}}{\llllll}\rrrr).$$
The parameter $\llllll$  determines the uniformity of the grid.

$\bullet$ {\it Quadratic non-uniform grid $\mathcal{M}^Q$ on the interval $\llll[s^{-},s^{+}\rrrr]$}
$$s_n=s^{-}+\llll( \frac{n}{N} \rrrr)^2 \llll(s^{+}-s^{-}  \rrrr),\quad h_n=\llll( \frac{2n+1}{N^2} \rrrr) \llll(s^{+}-s^{-}  \rrrr).$$
The Tavella-Randal and the quadratic meshes on the interval $[0,S]$ are defined with the functions
$$ \vvvv_Q(x)=S x^2, \quad \vvvv_{\llllll,K}(x)=K+\llllll \sinh \llll(c x+c_1\rrrr),$$
where
$$ c=\sinh^{-1}\llll(\dddd{S-K}{\llllll}\rrrr)+\sinh^{-1}\llll(\dddd{K}{\llllll}\rrrr),\quad
c_1=-\sinh^{-1}\llll(\dddd{K}{\llllll}\rrrr).
$$
\begin{figure}[h]
  \includegraphics{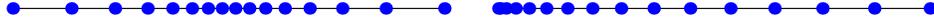}
 \caption{Graphs of the Tavella-Randall grid $\mathcal{M}^{TR}_4$ and the quadratic grid $\mathcal{M}^{Q}$ on the interval $[0,40]$ with $K=20$ and $N=15$ .}
\end{figure}

 The hyperbolic sine function is an odd increasing function with bounded first and second derivatives on the interval $[0,S]$. The inverse hyperbolic sine function is expressed with the natural logarithm function
$$\sinh^{-1} x = \ln \llll(x+\sqrt{x^2+1} \rrrr).$$
The function $\vvvv_{\llllll,K}(x)$ is equal to $K$ when $x=-c_1/c$ and   $\vvvv^\prime_{\llllll,K}(-c_1/c)=\llllll c$. The length of the smallest interval of the Tavella-Randall mesh is approximately
$\llllll c h=\llllll  h \llll(\sinh^{-1}\llll(\dddd{S-K}{\llllll}\rrrr)+\sinh^{-1}\llll(\dddd{K}{\llllll}\rrrr) \rrrr),$
$$\llllll c h=\llllll  h \llll(\ln\llll(\dddd{S-K}{\llllll}+\sqrt{1 +\llll(\dddd{S-K}{\llllll}\rrrr)^2}\rrrr)+\ln \llll(\dddd{K}{\llllll}+\sqrt{1 +\llll(\dddd{K}{\llllll}\rrrr)^2}\rrrr) \rrrr).$$
From the binomial and Taylor expansion formulas for $(1+y)^{1/2}$ and $\ln (1+y)$ we have

$\bullet$ When $\llllll$ is large
$$\llllll c h\approx \llllll  h \llll(\ln\llll(1+\dddd{S-K}{\llllll}\rrrr)+\ln \llll(1+\dddd{K}{\llllll}\rrrr) \rrrr)\approx \llllll h\llll(\dddd{S-K}{\llllll}+\dddd{K}{\llllll}\rrrr)=S h.$$

$\bullet$ When $\llllll$ is small
$$\llllll c h\approx \llllll  h \llll(\ln\dddd{2(S-K)}{\llllll}+\ln \dddd{2K}{\llllll}\rrrr)\approx \llllll h\llll( 4K(S-K)-2\ln \llllll\rrrr).$$
The function $-\llllll \ln \llllll\rightarrow 0$ when $\llllll\rightarrow 0$.  When $\llllll$ is large the Tavella-Randall mesh is almost uniform and when $\llllll$ is small the  mesh is highly non-uniform.

\subsection{Fourth-order compact approximation}
The central difference approximation for the second derivative has a second order accuracy on a uniform mesh. From the Taylor's expansion formula we can determine a second order accurate approximation \cite{SweilamRizkAbouHasan} for the second derivative on a three-point stencil of a non-uniform mesh which satisfies the conditions of Lemma 1. Now we determine a  compact approximation for the second derivative on a non-uniform mesh in the following form
\begin{equation}\label{CA}
d_{n}f_{n-1}'' + f_{n}'' +e_{n} f_{n+1}'' =a_{n} f_{n-1} +b_{n} f_{n} + c_{n}
f_{n+1}+E_n.
\end{equation}
From Taylor expansion at the point $s_n$, and setting the coefficients of $f_n,f'_n,f''_n,f'''_n$ and $f^{(4)}_n$ equal to zero we obtain a system of equations for the coefficients $a_n,b_n,c_n,d_n,e_n$,
\begin{equation*}
 \left\{\begin{aligned}
&a_n+b_n+c_n=0, \quad c_n h_n -a_n h_{n-1} =0,\quad \frac{a_n h_{n-1}^2+c_n h_n^2}{2}-d_n-e_n-1=0, \\
&\frac{c_n h_n^3-a_n h_{n-1}^3}{6}+d_n h_{n-1}-e_n h_n =0 \quad\frac{a_n h_{n-1}^4+c_n h_n^4}{24}-\frac{d_n h_{n-1}^2+e_n h_n^2}{2}=0.\\
       \end{aligned}
 \right.
\end{equation*}
Let $D_n=h_{n-1}^2+3  h_{n-1} h_n+h_n^2$. The system of equations  has the solution
$$a_n= \frac{12 h_n }{(h_{n-1}+h_n) D_n},\quad b_n= -\frac{12}{D_n},\quad
c_n= \frac{12 h_{n-1}}{(h_{n-1}+h_n) D_n},$$
$$d_n= \frac{h_n  \left(h_{n-1}^2+h_n h_{n-1}-h_n^2\right)}{(h_{n-1}+h_n)D_n},\quad e_n= \frac{h\llll(h_{n-1}^2+3 h_n +3 h_n  h_{n-1}+h_n^2  \rrrr)}{(h_{n-1}+h_n) D_n}.$$
 In the next lemma we show that approximation \eqref{CA} has a fourth-order accuracy on the class of non-uniform meshes determined by the functions $\vvvv\in C^2[0,1]$.
\begin{lem} Let $\vvvv$ be an increasing differentiable function on the interval $[0,1]$, with bounded first and second derivatives. Then the
compact approximation \eqref{CA} has fourth order accuracy on the non-uniform mesh $\mathcal{M}_\vvvv$.
\end{lem}
\begin{proof}From Taylor's formula the error $E_n$ of approximation \eqref{CA} is given by
$$E_n=\llll(\dddd{1}{6}\llll(d_n h_{n-1}^3 - e_n h_n^3 \rrrr)+\dddd{1}{120}\llll(c_n h_{n}^5 - a_n h_{n-1}^5 \rrrr)\rrrr)f^{(5)}_n.$$
Then
$$|E_n|=\llll|\dddd{h_{n-1}h_n\llll(h_{n-1}-h_n\rrrr)\llll(2h_{n-1}^2+5  h_{n-1} h_n+h_n^2\rrrr)}{h_{n-1}^2+3  h_{n-1} h_n+h_n^2}\dddd{f^{(5)}_n}{30}\rrrr|<h_{n-1}h_n\llll|h_{n-1}-h_n\rrrr|\dddd{\llll| f^{(5)}_n\rrrr|}{15},$$
$$|E_n|<\llll(s_n-s_{n-1}\rrrr)\llll(s_{n+1}-s_{n}\rrrr)\llll|s_{n+1}-2s_{n}+s_{n-1}\rrrr|\llll| f^{(5)}_n\rrrr|/15,$$
$$|E_n|<\llll(\vvvv\llll(x_n\rrrr)-\vvvv\llll(x_{n-1}\rrrr)\rrrr)\llll(\vvvv\llll(x_{n+1}\rrrr)-\vvvv\llll(x_{n}\rrrr)\rrrr)\llll|\vvvv\llll(x_{n+1}\rrrr)-2\vvvv\llll(x_{n}\rrrr)+\vvvv\llll(x_{n-1}\rrrr)\rrrr|\llll| f^{(5)}_n\rrrr|/15.$$
By the mean-value theorem there exist  $y_n\in(s_{n-1},s_n),z_n\in(s_n,s_{n+1})$ and $w_n\in(s_{n-1},s_{n+1})$ such that
$$|E_n|<\vvvv'\llll(y_n\rrrr)\vvvv'\llll(z_n\rrrr)\llll|\vvvv''\llll(w_n\rrrr) \rrrr|\dddd{\llll| f^{(5)}_n\rrrr|}{15}h^4<\dddd{1}{15}\llll(\max_{x\in [0,1]} \vvvv'(x) \rrrr)^2\llll(\max_{x\in [0,1]} \llll|\vvvv''(x)\rrrr| \rrrr)\llll(\max_{x\in [0,1]} \llll|f^{(5)}(x)\rrrr| \rrrr)h^4.$$
\end{proof}
The Tavella-Randall and the quadratic meshes are determined by the functions $ \vvvv_Q(x)=S x^2$ and $\vvvv_{\llllll,K}(x)=K+\llllll \sinh \llll(c x+c_1\rrrr)$. The two functions satisfy the requirements of Lemma 1. Therefore compact approximation \eqref{CA} has a fourth-order accuracy on the Tavella-Randall and the quadratic meshes. The requirements of Lemma 1 for the function $\vvvv$ are sufficiently general and include most of the non-uniform meshes used for numerical solution of differential equations.
\section{Compact finite-difference scheme for the time-fractional Black-Scholes equation}
In section 1 we outlined the main steps in the derivation of the fractional model for option prices \eqref{FM}. A detailed discussion of the model is given in \cite{LiangWangZhang}. The time-fractional Black-Scholes equation for European option prices is a special case of \eqref{FM} with $A_f=d_f=1$.
\begin{equation}\label{TFBS}
\frac{\partial^{\alpha} V} {\partial t^{\alpha} } +
\frac{1}{2} \sigma^{2} s^{2}
\frac{\partial^{2} V}{\partial s^{2} }  +(r-d) s
\frac{\partial V}{\partial s } -r V =0,
\end{equation}
In this section we determine a compact difference approximation for the TFBS equation for  European options  with payoff (final condition)
$V(s,T) = V^{\star} (s) =\max \{ K-s,0 \}$,
 where $K$ is the striking price. Additionally we prescribe Dirichlet boundary conditions $ V(0,t)=K, V(S,t) =0 $ on the bounded domain
$\Omega = [0,S]\times[0,T]$, where $0<K<S$.
For convenience  of the numerical construction, first we transform \eqref{TFBS} into an equivalent standard form
satisfying homogeneous Dirichlet boundary conditions.
Substitute
$$t:=T-t,\quad W(s,t):=V(s,t)+\dddd{K}{S}\llll(s-S\rrrr).$$
The function $W$ is a solution of the TFBS equation
	\begin{equation*}\label{W}
	\left\{
	\begin{array}{l}
\displaystyle{\dddd{\partial^{\alpha} W} {\partial t^{\alpha}}=\dddd{\sigma^{2}} {2} s^{2}
\dddd{  \partial^{2} W} {\partial s^{2}}+(r-d) s \dddd{\partial W} {    \partial s} -r W-\dddd{dK}{S}s+r K,}\\
\displaystyle{W(0,t)=W(S,t) =0, \; W(s,0)=W^{*}(s) =\max \{ K-s,0 \}+\dddd{K}{S}\llll(s-S\rrrr)}.
	\end{array}
		\right .
	\end{equation*}
In order to apply compact approximation \eqref{CA}, it is convenient to eliminate the convection term by substituting
$$
U(s,t) = s^{q}W(s,t), \quad\text{ where }\quad q= \dddd{r-d}{\ssss^2}.
$$
The function $U(S,t)$ is a solution of the TFBSD equation
\begin{equation}\label{TFBSD}
	\left\{
	\begin{array}{l}
\displaystyle{\frac{  \partial^{\alpha} U} {\partial t^{\alpha} }=
A s^{2}\dddd{\pppp^2 U}{\pppp s^2}+B U+F(s,t),}\\
\displaystyle{U(0,t)=U(s,t) =0, \;
U(s,0)=U^{*}(s) =s^{q}\llll(\max \{ K-s,0 \}+\dddd{K}{S}\llll(s-S\rrrr)
\rrrr)},
	\end{array}
		\right .
	\end{equation}
where
\begin{equation}\label{CS}
A=\dddd{\ssss^2}{2},\quad B=-\dddd{1}{2}\llll(r+d+q^2\ssss^2\rrrr),\quad
F(s,t)=-s^q \llll(\dddd{dK}{S}s-r K \rrrr).
\end{equation}
In the next  section  we construct a compact finite-difference scheme for the TFBSD equation using the fourth-order compact approximation \eqref{CA} on a three-point stencil of the Tavella-Randal and the quadratic non-uniform meshes and the $L1$-approximation for the Caputo fractional derivative defined as \cite{GaoSun}
$$\frac{\partial^{\alpha} U_n^m} {\partial t^{\alpha} }=\dddd{1}{\GGGG(2-\aaaa)\tttt^{\aaaa}}\sum_{k=0}^m \ssss_k^{(\aaaa)} U_n^{m-k}+E_n^m,
$$
where
$$\ssss_0^{(\aaaa)}=1,\quad \ssss_k^{(\aaaa)}=(k-1)^{1-\aaaa}-2k^{1-\aaaa}+(k+1)^{1-\aaaa},\quad \ssss_m^{(\aaaa)}=(m-1)^{1-\aaaa}-m^{1-\aaaa}.
$$
The $L1$-approximation has accuracy $O\llll(\tttt^{2-\aaaa} \rrrr)$ when $U(s,t)$ is a twice continuously differentiable function \cite{ChenXuZhu,GaoSun}. From the properties of the Caputo derivative, the TFBSD equation has a natural singularity at $t=0$.
The existence of a partial derivative of order $\aaaa$, where $0<\aaaa<1$ does not guarantee  that the integer-order partial derivatives of the function are continuous and bounded on the interval $[0,1]$. An important approach for analytical and numerical solution of linear and non-linear fractional differential equations is to use fractional power series.
In \cite{Dimitrov} we construct finite-difference schemes for the fractional sub-diffusion equation using the $L1$ and the modified $L1$-approximations for the Caputo derivative. In all numerical experiments the difference approximations have first order accuracy in the time direction. The same pattern is observed in the numerical solution of the TFBSD equation. The numerical test examples in Table 1 and Table 2 confirm that compact  difference approximation \eqref{DA} for the TFBSD equation has accuracy $O\llll(h^4+\tttt \rrrr)$.
\subsection{Compact difference approximation}
The form of the TFBSD equation is suitable  for using the fourth-order compact approximation \eqref{CA} on a non-uniform grid. Now  we construct a three-point compact finite-difference scheme for the TFBSD  on the non-uniform grid $\mathcal{G}_\vvvv$ of the rectangle $[0,S]\times [0,T]$ defines as
$$\mathcal{G}_\vvvv=\llll\{\llll(s_n,t_m\rrrr)|n=0,1,\cdots,N;m=0,1,\cdots,M\rrrr\},$$
where the points $s_n$ belong to a non-uniform mesh $\mathcal{M}_\vvvv$ of the interval $[0,S]$ and $t_m=m\tttt$, where $\tttt=T/M$ and $M$ is a positive integer. The finite-difference scheme uses the $L1$-approximation for the Caputo derivative and compact approximation \eqref{CA} for the second derivative. In the next section we compute the numerical solution of the TFBSD  equation  on the Tavella-Randall and the quadratic non-uniform grids.  By multiplying the TFBSD  equation by $1/s^2$ we obtain
$$\dddd{1}{s^2}\frac{  \partial^{\alpha} U} {\partial t^{\alpha} }=
A\dddd{\pppp^2U}{\pppp s^2}+\dddd{B}{s^2} U+H(s,t),\quad\text{where}\quad H(s,t)=\dddd{F(s,t)}{s^2}.$$
The function $U$ satisfies the following equations on a three-point stencil of the non-uniform grid $\mathcal{G}_\vvvv$,
\begin{align*}
\dddd{d_n}{s_{n-1}^2}\frac{\partial^{\alpha} U_{n-1}^m} {\partial t^{\alpha} }&=
A d_n \dddd{\pppp^2U_{n-1}^m}{\pppp s^2}+\dddd{B d_n}{s_{n-1}^2}U_{n-1}^m+d_n H_{n-1}^m,\\
\dddd{1}{s_{n}^2}\frac{\partial^{\alpha} U_{n}^m} {\partial t^{\alpha} }&=
A  \dddd{\pppp^2 U_{n}^m}{\pppp s^2}+\dddd{B}{s_{n}^2}U_{n}^m+H_{n}^m,\\
\dddd{e_n}{s_{n+1}^2}\frac{\partial^{\alpha} U_{n+1}^m} {\partial t^{\alpha} }&=
A e_n \dddd{\pppp^2 U_{n+1}^m}{\pppp s^2}+\dddd{B e_n}{s_{n+1}^2}U_{n+1}^m+e_n H_{n+1}^m.
\end{align*}
By adding the equations we obtain
$$\mathcal{C}_n^m=A\llll(d_n \dddd{\pppp^2 U_{n-1}^m}{\pppp s^2}+\dddd{\pppp^2U_{n}^m}{\pppp s^2}+e_n \dddd{\pppp^2U_{n+1}^m}{\pppp s^2}\rrrr)+B\llll(\dddd{d_n}{s_{n-1}^2}U_{n-1}^m+\dddd{1}{s_{n}^2}U_{n}^m +\dddd{e_n}{s_{n+1}^2}U_{n+1}^m \rrrr)+\mathcal{H}_n^m,$$
where
$$\mathcal{C}_n^m=\dddd{d_n}{s_{n-1}^2}\frac{\partial^{\alpha} U_{n-1}^m} {\partial t^{\alpha} }+\dddd{1}{s_{n}^2}\frac{\partial^{\alpha} U_{n}^m} {\partial t^{\alpha} }+\dddd{e_n}{s_{n+1}^2}\frac{\partial^{\alpha} U_{n+1}^m} {\partial t^{\alpha} },\quad
\mathcal{H}_n^m=d_n H_{n-1}^m+H_n^m+e_n H_{n+1}^m.$$
From compact approximation \eqref{CA}
$$\mathcal{C}_n^m=\llll(A a_n + \dddd{B d_n}{s_{n-1}^2}\rrrr)U_{n-1}^m+\left(A b_n+\frac{B}{s_n^2}\right)U_{n}^m +\left(A c_n+\frac{B e_n}{s_{n+1}^2}\right)U_{n+1}^m +\mathcal{H}_n^m+ O\llll(h^4 \rrrr).$$
By approximating the fractional derivative using the $L1$-approximation we obtain the systems of linear equations for the numerical solution of the TFBSD equation
\begin{equation}\label{DA}
Q \mathcal{U}^m=\mathcal{R}^m.
\end{equation}
where  $Q=\llll( q_{i,j}\rrrr)$ is a tridiagonal $(N-1)\times (N-1)$ matrix with elements
$$q_{n,n-1} = \dddd{d_n}{s_{n-1}^2}
 - \GGGG(2-\aaaa)\tttt^\aaaa\llll(A a_n + \dddd{B d_n}{s_{n-1}^2}\rrrr),\quad
q_{n,n}=\frac{1}{s_i^2}-\GGGG(2-\aaaa)  \tau ^{\alpha } \left(A b_n+\frac{B}{s_n^2}\right),$$
$$q_{n,n+1}=\frac{e_n}{s_{n+1}^2}-\GGGG(2-\aaaa)  \tau ^{\alpha } \left(A c_n+\frac{B e_n}{s_{n+1}^2}\right).$$
The right-hand side $\mathcal{R}^m=\llll( r^m_n\rrrr)$ of \eqref{DA} is an $(N-1)$-dimensional vector  with elements
\begin{align*}
r_n^m=\GGGG(2-\aaaa)& \tau ^{\alpha } \llll(d_n H(s_{n-1},t_m )+H(s_n,t_m )+e_n H(s_{n+1},t_m )\rrrr)\\
&-\frac{d_n}{s_{n-1}^2} \sum _{k=1}^m \sigma_k^{(\aaaa)} U^{m-k}_{n-1}-\frac{1}{s_n^2}\sum _{k=1}^m \sigma_k^{(\aaaa)} U^{m-k}_{n}-\frac{e_n}{s_{n+1}^2}\sum _{k=1}^m \sigma_k^{(\aaaa)} U^{m-k}_{n+1}.
\end{align*}
\subsection{Numerical experiments}
In he beginning of this section we showed that the TFBS equation for European option prices transforms to the TFBSD equation, where the coefficients $A$ and $B$ and the function $U(s,0)=U^*(s)$ are given by \eqref{CS}.
\begin{equation*}\label{W}
	\left\{
	\begin{array}{l}
\displaystyle{\frac{\partial^{\alpha} U} {\partial t^{\alpha} }=
A s^{2}\dddd{\pppp^2 U}{\pppp s^2}+B U+F(s,t),}\\
\displaystyle{U(0,t)=U(s,t)=0,\quad U(s,0)=U^*(s).}
	\end{array}
		\right .
	\end{equation*}
The TFBSD equation has a differentiable solution $U(s,t)=\left(1+2t+3t^2\right)\sin(\pi s)$ when
$$F(s,t)=\left(\frac{2t^{1-\alpha }}{\GGGG(2-\alpha)}+\frac{6t^{2-\alpha }}{\GGGG(3-\alpha)}\right) \sin (\pi s)+\left(A\pi ^2 s^2-B\right)\sin (\pi s)\left(1+2t+3t^2\right),$$
and
$$U(0,t)=U(1,t)=0,\;U^*(s)=\sin (\pi  s).$$
In the second columns of Table 1 and Table 2 we compute the orders of compact difference approximation \eqref{DA} with the above function $F(s,t)$ and  initial and boundary conditions on the quadratic and Tavella-Randall non-uniform grids  for  $\aaaa=0.75$ and $\aaaa=0.9$ and $A=1,B=2$ when $S=1$. The orders of numerical solution \eqref{DA}  of the TFBSD equation in the time and space directions when   $\sigma =0.1,r=0.08,d=0.025$ are given in the third columns of Table 1 and Table 2. The orders are computed by fixing one of the numbers $M=50$ and $N=50$ and computing the order of the numerical solution by doubling the value of the other number. In Figure 2 we compute the numerical solutions of the TFBS equation for European put options from the numerical solution of the TFBSD equation and the inverse transformations discussed in this section.
				\begin{table}[ht]
    \caption{Time and space orders of compact difference approximation \eqref{DA} for the TFBSD equation on the quadratic non-uniform grid $\mathcal{G}^{Q}$  and $\aaaa=0.75$.  }
      \centering
  \begin{tabular}{l c c }
  \hline
    $M(N=50)$ & $\quad \tttt-Order$ & $\quad \tttt-Order$  \\
		\hline
$100$   & $\quad 1.23372$          & $\quad 1.08047$ \\
$200$   & $\quad 1.23764$          & $\quad 1.06480$ \\
$400$   & $\quad 1.24169$          & $\quad 1.05343$ \\
$800$   & $\quad 1.24532$          & $\quad 1.04482$ \\
$1600$  & $\quad 1.24856$          & $\quad 1.03796$ \\
\hline
  \end{tabular}
      \centering
				\quad\quad
  \begin{tabular}{ l  c  c }
    \hline
$N(M=50)$   &$\quad h-Order$     &$\quad h-Order$  \\ \hline
$100$       &$\quad 4.01927$     &$\quad 3.95788$\\
$200$       &$\quad 4.00327$     &$\quad 3.98584$\\
$400$       &$\quad 4.00073$     &$\quad 3.99621$\\
$800$       &$\quad 4.00017$     &$\quad 3.99904$\\
$1600$      &$\quad 3.96748$     &$\quad 3.99952$\\
\hline
  \end{tabular}
\end{table}
				\begin{table}[ht]
    \caption{Time and space orders of compact difference approximation \eqref{DA} for the TFBSD equation on the Tavella-Randall non-uniform grid $\mathcal{G}^{TR}_6$  and $\aaaa=0.9$.}
      \centering
  \begin{tabular}{l c c }
  \hline
    $M(N=50)$ & $\quad \tttt-Order$ & $\quad \tttt-Order$  \\
		\hline
$100$   & $\quad 1.11894$          & $\quad 1.03640$ \\
$200$   & $\quad 1.10714$          & $\quad 1.03668$ \\
$400$   & $\quad 1.10250$          & $\quad 1.04803$ \\
$800$   & $\quad 1.10073$          & $\quad 1.04516$ \\
$1600$  & $\quad 1.10016$          & $\quad 1.03958$ \\
\hline
  \end{tabular}
      \centering
				\quad\quad
  \begin{tabular}{ l  c  c }
    \hline
$N(M=50)$   &$\quad h-Order$       &$\quad h-Order$  \\ \hline
$100$       &$\quad 4.02583$     &$\quad 3.97396$\\
$200$       &$\quad 4.00461$     &$\quad 3.98621$\\
$400$       &$\quad 4.00085$     &$\quad 3.99502$\\
$800$       &$\quad 4.00027$     &$\quad 3.99917$\\
$1600$      &$\quad 4.00003$     &$\quad 3.99972$\\
\hline
  \end{tabular}
\end{table}
\vspace{0cm}
\begin{figure}[ht]
  \includegraphics[width=0.38\textwidth]{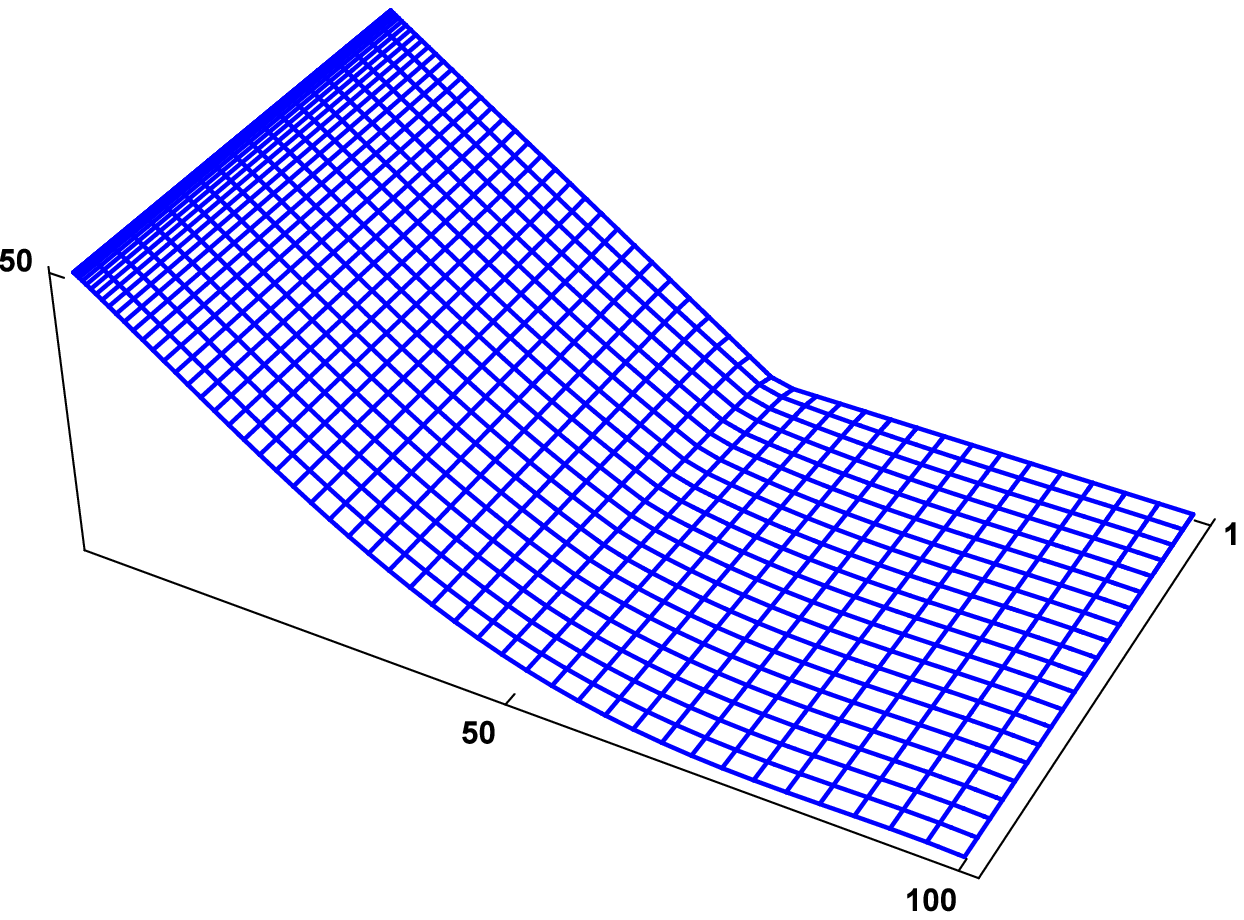}\hspace{1cm}
	  \includegraphics[width=0.38\textwidth]{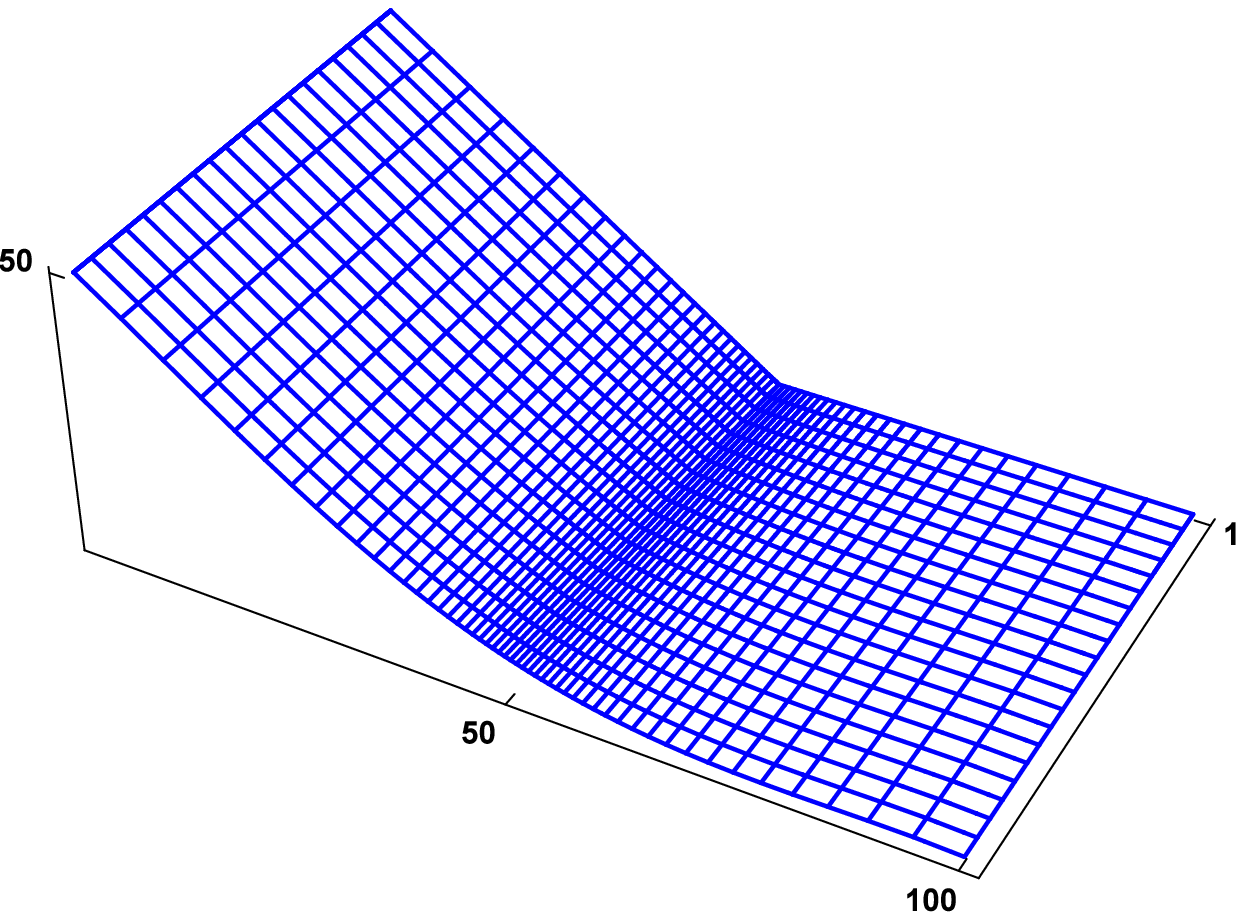}
 \caption{Numerical solutions of the TFBS equation for European put options  on the quadratic grid $\mathcal{G}^{Q}$ with $\aaaa=0.75$ and the Tavella-Randall grid $\mathcal{G}^{TR}_6$ with $\aaaa=0.9$  when $S=100,K=50$ and $N=50$.}
\end{figure}

When the TFBSD equation has a differentiable solution, the $L1$-approximation for the Caputo derivative has accuracy $O\llll(\tttt^{2-\aaaa} \rrrr)$. We can expect that for most functions $F(s,t)$ and initial condition $U(s,0)=U^*(s)$ the partial derivative $U_t(s,t)$ is unbounded at $t= 0$. This singularity of the TFBSD equation leads to a lower accuracy of the numerical solution in the time direction. The results of the numerical experiments presented in Table 1 and Table 2 are consistent with the expected  fourth-order accuracy in the space direction and first-order accuracy in the time direction of difference approximation \eqref{DA} for the TFBSD equation.
\section{Conclusions}
In the present paper we constructed a  compact difference approximation \eqref{DA} for the  TFBSD equation on a non-uniform spacial grid  which has a fourth-order accuracy in space. While the  accuracy of the numerical solution  is dominated by the accuracy in the time direction, difference approximation \eqref{DA} results in  a significant improvement in the computational time, since we use a much smaller number of subintervals in space. We discussed the numerical solution of the fractional model for European option prices when $A_f=d_f=1$. An important question for future work is to develop methods for numerical solution of the TFBS equation for other values of the parameters $A_f$ and $d_f$ and  accuracy in the time direction greater than $O(\tttt)$. In a forthcoming paper the convergence of the proposed method will be studied theoretically. Numerical solution of a new fractional nonlinear problem, corresponding to the integer Black-Scholes model, see e.g.\cite{K13, K14,K16} will be remained for our future consideration.
\section{Acknowledgments}
This research is supported by the European Union under Grant Agreement number 304617 (FP7 Marie Curie Action Project Multi-ITN STRIKE - Novel Methods in Computational Finance). The second author is also supported by Bulgarian National Fund of Science under Project I02/20-2014.

\end{document}